\documentclass[a4paper,11pt]{article}
\usepackage{amsmath,amsfonts}
\usepackage{theorem}
\usepackage{indentfirst}
\usepackage{graphicx}
\usepackage{longtable}
\usepackage{amssymb}
\usepackage{amsmath}
\usepackage{color}
\usepackage{setspace}
\usepackage{amsfonts}%
\usepackage{algorithm}
\usepackage{algorithmic}
\usepackage[round]{natbib}
\usepackage[driverfallback=dvipdfm]{hyperref}
\providecommand{\U}[1]{\protect\rule{.1in}{.1in}}

\theoremstyle{plain} { \theorembodyfont{\rmfamily}
	\newtheorem{definition}{Definition}

}
\newtheorem{proposition}{Proposition}
\newtheorem{lemma}[proposition]{Lemma}

\newtheorem{theorem}{Theorem}

\newcommand{\R}{\mathbb{R}}
\newcommand{\E}{\mathbb{E}}

\newcommand{\var}{\text{\rm Var}}

\newcommand{\qed}{\hfill \mbox{\raggedright \rule{.07in}{.1in}}}
\newenvironment{proof}{\vspace{1ex}\noindent{\bf Proof}\hspace{0.5em}}
{\hfill\qed\vspace{1ex}}

\usepackage[left=1in,right=1in,top=1in,bottom=1in]{geometry}
\usepackage{bm}
\begin{document}
\title{Na\"ive Markowitz Policies\thanks{The first version of the paper was completed in 2017, and part of the results was included in the first author's PhD thesis defended in 2020. The paper was finalized when the second author was on vacation in Las Vegas, a place arguably  ideal for observing the ``na\"ive" behaviors studied in the paper.} }
\author{Lin Chen\thanks{Department of Industrial
Engineering and Operations Research, Columbia University, New York, New York 10027, USA, \texttt{lc3110@columbia.edu}}\ \ \ Xun Yu Zhou\thanks{Department of Industrial Engineering and Operations Research, Columbia University, New York, New York 10027, USA, \texttt{xz2574@columbia.edu}}}
\maketitle

\begin{abstract}
We study a continuous-time Markowitz mean--variance portfolio selection model in which a na\"ive agent, unaware of the underlying  time-inconsistency, continuously reoptimizes over time.
We define the resulting na\"ive policies through the limit of discretely na\"ive policies that are committed only in very small time intervals, and  derive them analytically and  explicitly. We compare na\"ive policies with pre-committed optimal policies and with consistent planners' equilibrium policies in a Black--Scholes market, and find that the former are mean--variance inefficient starting from any given time and wealth, and always take  riskier exposure than equilibrium policies.

\bigskip

\end{abstract}

\section{Introduction}
The Markowitz mean--variance (MV) portfolio selection model (\citealp{MH52} and \citealp{MH59}) is a monumental work in quantitative finance. The model formulates the investment problem as striving to achieve the best balance between return and risk, represented respectively by the mean and variance of the final portfolio worth. Its variants, extensions and implications have been passionately  studied in theory and applied in practice to this day. 

The original MV model is formulated for a static single period and solved by quadratic program. It is natural and necessary to  extend it to the dynamic setting, both in discrete time and continuous time. However,  a dynamic MV model is inherently {\it time inconsistent}; namely, any ``optimal" policy for the present moment  will generally not be optimal for the next moment.\footnote{Here a ``policy" is a {\it plan} that maps any given time and state to an action (a portfolio in the MV model).
It is also called a {\it feedback control  law} in control theory.} This inconsistency comes from the variance term that does not satisfy the tower rule: unlike the mean, there is no consistency over time in evaluating the same variance of the final wealth.
As a result, in sharp contrast to the classical time-consistent models, there is no such notion as a {\it dynamically optimal policy} for a time-inconsistent model because any such policy, once planned for this moment,  may need to be given up quickly (and  {\it instantly} in a continuous-time setting) in favor of a different plan  at the next moment. Technically, time-inconsistency poses fundamental challenges in ``solving" -- whatever ``solving"  means -- the problem because the Bellman optimality principle, which is the very foundation of the classical dynamic programming for studying dynamic optimization problems, is no longer valid.

Economists have recognized and studied time-inconstancy since as early as the 1950s. The foundational  paper \cite{SR56}  {\it describes} three types of agents when facing time inconsistency. Type 1, a ``na\"ivet\'e" (or na\"if), is unaware of the time inconsistency and at any given time and  state of affairs seeks an ``optimal" policy  for that moment only, without knowing that he will not uphold that policy for long. As a result, his policies  change all the times, and the eventual policy that is being actually carried out can be vastly and characteristically different from any of his short-lived ``optimal" policies he originally planned to execute.\footnote{For instance, \cite{casino} shows, in a casino gambling model (which is time-inconsistent in discrete time due to probability weighting),
a na\"ive gambler's initial plan was to gamble as long as possible when winning but to stop if he started accumulating losses, he actually ends up doing
the {\it opposite}: he gambles as long as possible when losing and stops once
he accumulates some gains. Similar behaviors are also observed, and indeed prevalent, in stock investment especially with retail investors.}   The next two types realize the issue of time inconsistency
but act differently. Type 2 is a ``precommitter" who solves the optimization problem only at time 0 and sticks  to the resulting policy throughout (via some ``commitment device" if necessary and available), recognizing  that the original policy may no longer be optimal at later times. Type 3 is a ``consistent planner"  who is unable to precommit and realizes that her future selves may abandon  whatever plans she makes now. Her resolution is to optimize taking the future deviations from the current plan  as {\it constraints}, effectively leading to a game among selves at different times.
The resulting policies are called equalibrium ones.

It is important to note that it is not meaningful to determine
which type is superior than the others, simply because there are no uniform criteria
to compare them. In this sense, the Strotzian approach to time-inconsistency is both {\it normative} (i.e. to advise people about the best course of actions,  especially in Types 2 and 3) and {\it descriptive} (i.e. to describe what people are actually  doing, as more with Type 1). 

Mathematically, model formulations and solutions for deriving the
three types of agent policies call for different treatments as they are very different from each other. The problems are also challenging due to the invalidity of the dynamic programming approach. In the last decade, there have been significant developments in studying time-inconsistent models analytically, mainly in three different settings:
MV portfolio selection, and optimization problems involving  non-exponential
discounting or probability weighting; see \cite{survey} for a recent survey on the related works.  For the MV models, earlier works focused on Type 2, pre-committed agents;
see, e.g., \cite{R89,H71,LN00,ZL00,LIMZ02,BJPZ05,X05,LZ06}, although most of these works did not spell out that their solutions were pre-committed ones. Later research gradually shifted to Type 3, consistent
planners; see, e.g. \cite{BC05, HJZ12,  BA14, BMZ14,HJ17}.

In contrast to the rich literature on pre-committed agents and consistent planners, there are far fewer works on the general behaviors of na\"ive agents, and almost none  in continuous time (not necessarily limited to MV models). \cite{casino, HOZ22} study na\"ive strategies in casino gambling models  which are inherently discrete time. As shown in these papers, finding na\"ive policies in discrete time is rather straightforward technically if the pre-committed polices are already available: at each discrete time point one solves and obtains the corresponding pre-committed policy, holds it until the next time point when one re-solves the pre-committed problem, and repeats these steps until the terminal time. The {\it eventual} na\"ive policy is then just to ``paste" these piece-wise pre-committed policies together. This pasting approach, however, does not work for the continuous-time setting. Indeed, at each given time and state, say $(s,y)$,  a pre-committed policy is executed and {\it instantaneously} discarded, while a policy applied for just one single time--state initial point $(s,y)$ has no impact on the dynamics in continuous time. As a result, it is unclear how to paste these continuously changing policies and, even one found a way to do it, how to interpret the resulting policy.

We address this issue specific to continuous time  and make two main contributions in this paper. First, to our best knowledge we are  the first to define precisely the na\"ive policies in the original sprit of \cite{SR56}  but adapted to the continuous-time setting, premised upon the notion that any continuous-time behavior is the limit of discrete-time behaviors when the time-step approaches zero.\footnote{An analogy here is that the Brownian motion is just the limit of a simple random walk when the step size diminishes  to zero.} We fix a set of discrete time points and consider a fictitious agent who only optimizes at each of these points and holds the resulting pre-committed policy until the next point. It is then natural to use the ``limit" -- in a certain sense -- of these discretely na\"ive agents when the step size becomes asymptotically small to  describe the na\"ive behavior in the original continuous-time model. One technical subtlety here is that policies are generally only  measurable functions whose limit is difficult to analyze.
We consider instead the limiting process of the wealth processes -- which are analytically better behaved -- of those discrete agents, and find the policy that generates this limiting process as the wealth process.
A main advantage of our approach is that it is both general and constructive. It is general because the definition of a na\"ive policy applies readily to any time-inconsistent problems beyond MV (see \citealp{CZ20b} for an extension to the general stochastic linear--quadratic  control problem), and it is constructive because the definition itself points to the direction of {\it deriving}  a na\"ive policy.

The second contribution is to compare the na\"ive policies with the other Strotzian types of policies. Be mindful that it does not make much sense to use either mean or variance of the  terminal wealth alone for comparison, as the essence of the MV model is to achieve a best trade-off between the two criteria. Instead, MV efficiency ought to be the primary criterion.
Between a na\"ivet\'e and a pre-committer, starting from any  given point of time and state, the latter is MV efficient by definition while we show that the former is not (although he elevates the expected terminal wealth than he originally planned).
To compare
na\"ive  and equilibrium policies which are both MV inefficient, we use an objective metric which  is the risky weight defined as the fraction of dollar amount invested  in stocks.
We show that a na\"ivet\'e {\it always} allocate strictly higher risky weight than the two types of consistent planners considered by \cite{BMZ14} and \cite{HJ17} respectively. This in turn suggests that the na\"ive policies tend to be more risk-taking than their  consistent planning counterparts.\footnote{An analogous result is proved in \cite{HOZ22} for a casino gambling model: a na\"ive gambler stops gambling no earlier than a gambler doing consistent planning.} 

\cite{PP17} introduce the notion of ``dynamic optimality" in a continuous-time MV model, which seems to bear some relevance to
na\"ive policies (although the paper stops short of commenting on it). Definition 2 therein defines  a dynamically optimal policy  as there being no other policy applied at present
time could produce a more favourable value at the terminal time. However, as discussed earlier, in a time-inconsistent problem there is no such thing as ``dynamic optimality": as much as a na\"ivet\'e attempts to reoptimize continuously over time, the resulting actual policy at any given time may significantly deviate from the pre-committed optimal one (and therefore is MV inefficient, and indeed {\it not} optimal in any sense). On the other hand, \cite{PP17} conjecture the analytical formula of such a ``dynamically optimal" policy for a single stock Black--Scholes market without explaining where it comes from. Hence the solution method is {\it ad hoc} and it is unclear whether the existence of such a policy is prevalent and, if yes, how to extend the conjecture  to  a more general MV setting (e.g., one with more than one risky asset) or to  other time-inconsistent problems (e.g. with non-exponential discounting or probability weighting). By contrast, our definition of na\"ive policies is general and our derivation of these policies is constructive.

The rest of the paper is organized as follows. In Section \ref{Naive_problem_formulation} we formulate the continuous-time MV portfolio selection model. In Section \ref{Naive_analysis} we introduce the so-called $2^{-n}$-committed policies, which are commited only during a small interval of length $2^{-n}$, before reoptimization. We consider the limit of the wealth processes under these policies as $n\to\infty$, and define the policy that generates this limiting wealth process as a na\"ive policy. We then state the main result that expresses na\"ive policies analytically. In Section \ref{Naive_comparisons} we compare na\"ive policies with other types of policies in a Black--Scholes market. Section \ref{Naive_conclusions} concludes the paper. Proofs related to the main result are placed in Appendices.

\section{A Continuous-Time Markowitz Model}\label{Naive_problem_formulation}
In this section we review the continuous-time Markowitz MV model.  We first introduce notations.

Throughout this paper, $M^{\top}$ denotes the transpose of any vector or matrix $M$,
while all vectors are {\it column} vectors unless otherwise specified.
A fixed filtered complete probability space
$(\Omega,\mathcal{F},\mathbb{P},\{\mathcal{F}_t\}_{t\geq 0})$ is given along with a standard $\{\mathcal{F}_t\}_{t\geq 0}$-adapted, $m$-dimensional Brownian motion $W(t)\equiv (W^1(t),...,W^m(t))^{\top}$. We use $f$ or $f(\cdot)$ to denote the {\it function} $f$, and $f(x)$ to denote
the {\it function value} of $f$ at $x$. Likewise, we use $X$ or $X(\cdot)$ to denote a stochastic process $X=\{X_{s},$ $s\geq 0\}$. Given a Hilbert space $H$ and $b>a\geq0$,
we denote by
$L^2([a,b];H)$  the Hilbert space of $H$-valued, square-integrable functions $f$ on $[a,b]$ endowed with the norm $(\int_a^b ||f(t)||^2_H dt)^{1/2}$. Moreover,   we denote by $L^2_\mathcal{F}([a,b];\R^m)$ the Hilbert space  of  $\R^m$-valued,  square-integrable and  $\{\mathcal{F}_t\}_{t\geq 0}$-adapted stochastic processes $g$ endowed with the norm  $\left[\E\int_a^b ||g(t)||^2 dt\right]^{1/2}$, where  $||\cdot||$ is the $L^2$ norm in a Euclidean space.

A financial market has $m+1$ assets being traded continuously. One of the assets is a bank account whose price process $S_0$ is subject to the following equation:
\begin{equation}
dS_0(t)=r(t)S_0(t)dt,\ t\geq0;\ S_0(0)=s_0>0,
\end{equation}
where the interest rate function $r(\cdot)$ is deterministic. The other $m$ assets are stocks whose price processes $S_i,i=1,...,m$, satisfy the following stochastic differential equations (SDEs):
\begin{equation}
dS_i(t)=S_i(t)\left[b_i(t)dt+\sum\limits_{j=1}^m \sigma_{ij}(t)dW^j(t)\right],\ t\geq0;\;\; S_i(0)=s_i>0,
\end{equation}
where $b(\cdot)$ and $\sigma_{ij}(\cdot)$, the appreciation and volatility rates functions respectively, are scalar-valued and deterministic.
Set the excess rate of return vector function and the volatility matrix function respectively as
$$B(t):=(b_1(t)-r(t),...,b_m(t)-r(t))^{\top},\;\;\sigma(t):=(\sigma_{ij}(t))_{m\times m}.$$

An agent 
has total {\it wealth} $X(t)$ at time $t\in[0,T]$, where $T$ is a given terminal time of the investment horizon.  Assuming that the trading of shares takes place in a self-financing fashion and that there are no transaction costs,  the process $X$ satisfies the {\it wealth equation}
\begin{equation}\label{wealtheq}
dX(t)=\left[r(t)X(t)+B(t)^\top\pi(t)\right]dt+\pi(t)^\top\sigma(t)dW(t),\ t\in[0,T],
\end{equation}
where each $\pi_i(t),i=1,2,...,m$, denotes the total market value of the agent's wealth in the $i$-th asset, resulting in a {\it portfolio} $(\pi_1(t),...,\pi_m(t))^\top$, at time $t$. The agent considers portfolio choice at time $s$ when her wealth is $y$, where $(s,y)\in [0,T)\times\R$ is given.
The process $\pi\equiv(\pi_1,...,\pi_m)^\top=\{\pi(t): s\leq t\leq T\}$ is called an {\it admissible portfolio} (process) for $(s,y)$ if $\pi\in L^2_\mathcal{F}([s,T];\R^m)$ and the wealth equation (\ref{wealtheq}) with initial condition $X(s)=y$  admits a unique strong solution. Denote by ${\cal U}_{s,y}$ the set of admissible portfolio processes for $(s,y)$.

We focus on a portfolio {\it policy} $\bm{\pi}=\bm{\pi}(\cdot,\cdot)$ which is a {\it deterministic} map from $[0,T]\times\R$ to $\R^m$. Such a policy specifies a portfolio $\bm{\pi}(t,x)$ when  time is $t$ and wealth is $x$.\footnote{In control theory, the policy here is also called the {\it feedback} control law, whereas the portfolio process corresponds to the {\it open-loop} control.}
In the classical,  time-consistent setting, a policy $\bm{\pi}(\cdot,\cdot)$ is  independent of the initial time--state pair $(s,y)$, meaning that it is implemented no matter when and where one starts. Such policies are called {\it time-consistent} ones. A time-consistent policy $\bm{\pi}=\bm{\pi}(\cdot,\cdot)$ is called admissible if for {\it any} $(s,y)\in [0,T)\times\R$, the following SDE obtained by substituting $\bm{\pi}$ into
the wealth equation (\ref{wealtheq})
\begin{equation}\label{wealtheqfb}
dX(t)=\left[r(t)X(t)+B(t)^\top\bm{\pi}(t,X(t))\right]dt+\bm{\pi}(t,X(t))^\top\sigma(t)dW(t),\ t\in[0,T];\;\; X(s)=y,
\end{equation}
admits a unique strong solution $X$ and, moreover, the resulting portfolio process $\pi\in {\cal U}_{s,y}$ where $\pi(t):=\bm{\pi}(t,X(t))$, $t\in[s,T]$. Note  that the wealth--portfolio {\it process} pair $(X,\pi)$ {\it depends on} the initial  $(s,y)$, and we say $(X,\pi)$ is {\it generated} from the policy $\bm{\pi}$ with respect to $(s,y)$. 

The classical verification theorem for time-consistent problems (e.g. \citealp{YZ99}) dictates that, under standard assumptions, there exists a time-consistent policy that generates optimal wealth--portfolio process pair $(X,\pi)$ for any given initial  $(s,y)$.

\smallskip

The following assumptions are in force throughout this paper.

\textbf{(A1)} $r(t), B(t)$ and $\sigma(t)$ are uniformly bounded on $ [0,T]$.

\textbf{(A2)} $B(t)\neq 0$ a.e.$t\in[0,T]$ and $\sigma(t)\sigma(t)^\top\geq \delta I, \forall t\in [0,T]$
for some $\delta>0$. 

\smallskip

 Given $(s,y)\in [0,T)\times\R$, the Markowitz mean--variance portfolio selection  problem over $[s,T]$ is
\begin{equation}\label{prob1}
\min\limits_{\pi(\cdot)\in {\cal U}_{s,y}}\ \var_{s,y}(X(T))
\end{equation}
\begin{equation}\label{ft0}
\text{subject to}\
\begin{cases}
\E_{s,y}[X(T)]=yf(s,T),\\
(X(\cdot),\pi(\cdot))\ \text{satisfy}\ (\ref{wealtheq}) \mbox{ with } X(s)=y
\end{cases}
\end{equation}
where $\var_{s,y}$ and $\E_{s,y}$ denote respectively the variance and expectation conditional on $\mathcal{F}_s$ and $X(s)=y$, and
$f(u,v), 0\leq u\leq v\leq T$, is a given deterministic real-valued function satisfying $f(u,u) = 1, \forall u\in[0,T]$.  The number $f(u,v)$ represents  the desired  growth factor over the time horizon $[u,v]$. It is economically sensible to consider the expected mean target to be dependent of the initial $(s,y)$, which is equivalent to the state-dependend risk aversion considered in \cite{BMZ14}. \cite{HJ17} consider a more general target $L(s,y)$ instead of $yf(s,T)$; see also Section 4 of this paper.

We add an assumption on $f$ throughout this paper:

\textbf{(A3)}  $f\in C^1([0,T]\times[0,T])$, $f(u,v) \geq e^{\int_u^v r(t) dt}$, $\forall\ 0\leq u \leq v \leq T$, and $-\infty<\frac{\partial f}{\partial t}(t,T)|_{t=T}<\infty$.

\noindent The second part of this assumption is natural, demanding  the target return to be  at least as great as  the risk-free return.

Given $(s,y)$, the relation between $\var_{s,y}(X_*(T))$ and $\E_{s,y}[X_*(T)]$, where $X_*(T)$ is the optimal terminal wealth of the
problem (\ref{prob1}) -- (\ref{ft0}), is called an {\it efficient frontier} with respect to $(s,y)$, which gives the best risk--return tradeoff for future investment when standing at $(s,y)$.

The
problem (\ref{prob1}) -- (\ref{ft0}) has been solved explicitly in literature; see e.g. \cite[Theorem 2.1]{LZ06},\footnote{The previous results such as \cite[Theorem 2.1]{LZ06} are for the case when $(s,y)=(0,x_0)$, but they extend readily to arbitrary initial
$(s,y)$ because the underlying mathematical problem of the latter is the same.}
with the following {\it unique} optimal policy (conditional on $\mathcal{F}_s$ and $X(s)=y$)
\begin{equation}\label{us}
\begin{aligned}
\bm{\pi}_*(t,x;s,y)=-[\sigma(t)\sigma(t)^\top]^{-1}B(t)^\top \left[x-{\gamma}(s,T) e^{-\int_t^T r(v) dv}y\right],\;\;(t,x)\in [s,T)\times \R,
\end{aligned}
\end{equation}
where
\begin{equation}\label{gamma}
{\gamma}(s,T):=
\frac{f(s,T)-e^{\int_s^T [r(v)-\rho(v)] dv}}{1-e^{-\int_s^T \rho(v) dv}},\;\;s\in[0,T),
\end{equation}
with
$$\rho(t):=B(t)[\sigma(t)\sigma(t)^\top]^{-1}B(t)^\top>0.$$
Note that l'H\^ospital's rule along with Assumptions {\bf (A2)-(A3)} yield
that $\gamma(\cdot,T)$ is continuous at $T$; hence is uniformly bounded on  $[0,T]$.

Substituting the policy (\ref{us}) into the wealth equation (\ref{wealtheq}) we obtain that the corresponding optimal wealth process is determined by the following SDE:
\begin{equation}\label{Ws}
\begin{cases}
dX_*(t)=\left[(r(t)-\rho(t))X_*(t)+{\gamma}(s,T)\rho(t) e^{-\int_t^T r(v) dv}y\right]dt \\
\ \ \ -B(t)(\sigma(t)\sigma(t)^\top)^{-1}\sigma(t)\left[X_*(t)-{\gamma}(s,T)e^{-\int_t^T r(v) dv}y\right]dW(t),\;\;t\in[s,T], \\
X(s)=y.
\end{cases}
\end{equation}

Finally, the efficient frontier at $(s,y)$ is
\begin{equation}\label{ef}
\var_{s,y}(X_*(T))=\frac{1}{e^{\int_s^T \rho(v) dv}-1}\left(\E_{s,y}[X_*(T)]
-ye^{\int_s^T r(v) dv}\right)^2.
\end{equation}

In sharp contrast to the time-consistent setting,
the policy $\bm{\pi}_*(\cdot,\cdot;s,y)$ given by (\ref{us}) now depends on the initial pair  $(s,y)$
{\it explicitly}. 
If the agent sticks to this policy during the entire future time period $[s,T]$ without subsequently altering it, then it is the so-called optimal {\it pre-committed} policy. If the agent is na\"ive {\it \`a la} Strotz who reoptimizes at every subsequent time moment, then the policy (\ref{us}) will be abandoned immediately (indeed instantaneously) at any $\tilde s>s$. More precisely, suppose the agent carries out (\ref{us}) for a (little) while and reaches the state $X(\tilde s)$ at time $\tilde s>s$. Now the {\it current} initial time becomes $\tilde s$ and the {\it current} initial state is $X(\tilde s)$. If the agent reoptimizes the problem for the remaining duration $[\tilde s,T]$, then the corresponding policy at $(\tilde s,X(\tilde s))$ is (conditional on $\mathcal{F}_{\tilde s}$)
\begin{equation}\label{ut}
\begin{aligned}
\bm{\pi}_*(t,x;\tilde s,X(\tilde s))=-[\sigma(t)\sigma(t)^\top]^{-1}B(t)^\top \left[x-{\gamma}(\tilde s,T) e^{-\int_t^T r(v) dv} X(\tilde s)\right],\;\;(t,x)\in [\tilde s,T)\times \R.
\end{aligned}
\end{equation}
Clearly, the two policies (\ref{us}) and (\ref{ut}) are generally different
as two {\it functions} on $[\tilde s,T]\times \R$.

So,  problem (\ref{prob1}) -- (\ref{ft0}) admits a policy (\ref{us}) that is optimal for the current $(s,y)$ {\it only}. 
In other words, the pre-committed optimal policy depends inherently on $(s,y)$, which in turn causes the time-inconsistency of the policy and hence that of the problem, as discussed above. A time-inconsistent policy of the type (\ref{us})  is defined  only for the given $(s,y)$. 


\section{Na\"ive Policies}\label{Naive_analysis}


A na\"ivet\`e (``he") always ``reoptimizes"  under current information; as a result he devises policies and then instantly abandons them in the  continuous-time setting. Although at each given time he tries to follow the pre-committed optimal policy (\ref{us}) but his {\it eventual} policy due to the constant changes could be completely different from (\ref{us}). In this section, we define na\"ive policies rigorously, and then derive them in analytical form for the MV problem (\ref{prob1})--(\ref{ft0}).

\subsection{A $2^{-n}$-committed agent}
As discussed earlier, the difficulty of defining and analyzing na\"ive policies lies in the continuous-time setting of the problem. We overcome this difficulty  by introducing an auxiliary agent, named the {\it $2^{-n}$-committed agent}, to approximate the behavior of the na\"ivet\`e.

A $2^{-n}$-committed agent (``she") is one who behaves ``in between" a pre-committer and a na\"ivet\`e. Specifically, she partitions the time horizon $[0,T]$ into $2^n$ equal-length intervals, with the partitioning points being $\{t_k\}_{k = 0}^{2^n}$ where $t_k = \frac{kT}{2^n}$. She first solves problem (\ref{prob1})-(\ref{ft0}) with $(s,y)=(0,x_0)$ to obtain the pre-committed optimal policy
$\pi(\cdot,\cdot;0,x_0)$ defined by (\ref{us}). She implements and commits to this policy until time $t_1$ when her wealth becomes $X(t_1)$, at which she resolves problem (\ref{prob1})-(\ref{ft0})  with $(s,y)=(t_1,X(t_1))$ and switches to  the policy $\pi(\cdot,\cdot;t_1,X(t_1))$. She commits to this new policy  until $t_2$ before changing it to $\pi(\cdot,\cdot;t_2,X(t_2))$. She then repeats these steps  until time $T$. Figure \ref{multimePlot} illustrates the resulting wealth process under this construction.
 \begin{figure}[ptb]
 	\centering
 	\includegraphics[width=6.0in,height=3in]{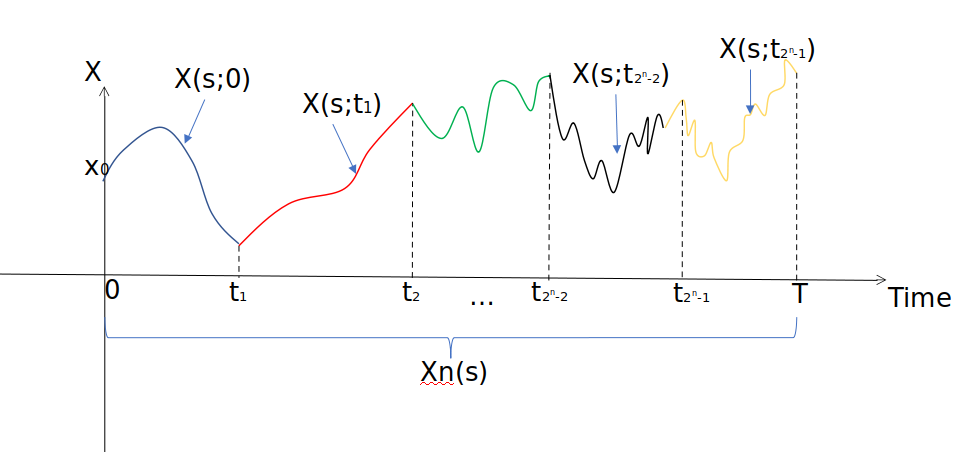}
 	\caption{This figure shows a sample path of the wealth process $X_n(\cdot)$ of the $2^{-n}$-committer. Each segment of the process, represented by a different color,
 follows the pre-committed optimal policy devised at the beginning of the corresponding time interval.  The wealth process is continuous. }\label{multimePlot}
 \end{figure}

Denote by $\{X_*(t;t_k): t\in[t_k,t_{k+1}]\}$ the above wealth process in the time interval $[t_k,t_{k+1}],\;k=0,1,\cdots,2^{n-1}$, with $X_*(0;0)=x_0$. By (\ref{Ws}), these processes $X_*(t;t_k),t\in[t_{k},t_{k+1}]$, $k=0,1,\cdots,2^{n-1}$, can be determined by the following SDEs recursively:
\begin{equation}\label{rec}
\begin{cases}
dX_*(t;t_k)=\left[(r(t)-\rho(t))X_*(t;t_k)+\gamma(t_k,T)\rho(t)  e^{-\int_t^T r(v) dv}X_*(t_k;t_{k - 1})\right]dt  \\
\ \ \ -B(t)(\sigma(t)\sigma(t)^\top)^{-1}\sigma(t)\left[X_*(t;t_k)-\gamma(t_k,T) e^{-\int_t^T r(v) dv}X_*(t_k;t_{k - 1})\right]dW(t),\;\; t\in[t_k,t_{k+1}], \\
X_*(t_k;t_k)=X_*(t_k;t_{k-1}),
\end{cases}
\end{equation}
where $X_*(t_0;t_{-1})$ is defined as $x_0$.

Now, by ``pasting" $X_*(\cdot;t_k)$, $k = 0,1,...,2^n-1$, we obtain the following process:
\begin{equation}\label{pasting}
X_n(s):=
\begin{cases}
X_*(s;0),& \mbox{$0\leq s< t_1$}, \\
X_*(s;t_1),& \mbox{$t_1\leq s< t_2$},\\
...\\
X_*(s;t_{2^n-1}),& \mbox{$t_{2^n-1}\leq s\leq T$},\\
\end{cases}
\end{equation}
which is the wealth process of the $2^{-n}$-committed agent, visualized by  Figure \ref{multimePlot}. Obviously, this process is adapted and continuous on $[0,T]$.

\subsection{Na\"ive policies}

While this $2^{-n}$-committed agent behaves somewhere between a pre-committed agent and a na\"ive one, she is closer   to the latter when $n$ becomes larger. Therefore, we define a na\"ive policy through the limit (in certain sense) of the  $2^{-n}$-committed wealth process as $n\rightarrow\infty$.

\begin{definition}
If the $2^{-n}$-committed wealth process $X_n$ converge to an adapted process $X$ in some sense, and the limiting process $X$ can be generated by a {\it time-consistent} admissible policy $\bm{\pi}^*=\bm{\pi}^*(\cdot,\cdot)$, then $\bm{\pi}^*$ is called a {\it na\"ive policy} of the problem (\ref{prob1})-(\ref{ft0}).
\end{definition}

Some remarks on this definition are in order. First, this definition applies  to more general time-consistent problems instead of just the current Markowitz problem.
As such, we intentionally leave vague
the precise sense in which $X_n$ converge to $X$ in order to make the definition  general and applicable to other problems. For the present problem, we will see momentarily that the convergence is in the weak-$L^2$ sense. Second, a  na\"ive policy in itself must be time-consistent, meaning that it can no longer depend on any initial $(s,y)$ and, in particular, on $(0,x_0)$, even though each $X_n$ is indeed constructed starting from a specific pair $(0,x_0)$. Finally, we do not define a na\"ive policy as simply the limit of $2^{-n}$-committed policies, because policies are in general only measurable and they may not converge and are hard to analyze. Instead, we consider the limit of  wealth processes that are much better behaved, and then use the limiting wealth equation to recover the corresponding na\"ive policy.

The following proposition, whose proof is deferred to Appendix \ref{proof_normfinite}, indicates that the $2^{-n}$-committed wealth processes $X_n$, $n=1,2,\cdots$, are uniformly bounded in  $L^2_\mathcal{F}([0,T];\R)$.
\begin{proposition}\label{normfinite}
	It holds that   $$||X_n||^2 := \E\int_0^T |X_n(s)|^2 ds < \infty,\ \forall n.$$
	Moreover, $||X_n||^2$ is uniformly bounded in $n$.
\end{proposition}

 Due to Proposition \ref{normfinite}, the sequence $\{X_n\}_{n=1}^\infty$ is uniformly bounded in the Hilbert space $L^2_\mathcal{F}([0,T];\R)$, and hence is weakly compact. So there exists a weakly convergent subsequence (still denoted as $\{X_n\}_{n=1}^\infty$ without loss of generality) and a process $X\in L^2_\mathcal{F}([0,T];\R)$ such that
$$X_n\rightarrow X\ \text{weakly in $L^2_\mathcal{F}([0,T];\mathbb{R})$}.$$
The following theorem is the main result of the paper, which characterizes this limiting process and, consequently, the na\"ive policy.

\begin{theorem}\label{theorem1}
The weakly limiting process $X$ satisfies the following SDE:
\begin{equation}\label{XSDE}
\begin{cases}
dX(t)=\left[(r(t)-\rho(t))+{\gamma}(t,T)\rho(t) e^{-\int_t^T r(s) ds}\right]X(t)dt \\
\ \ \ -B(t)(\sigma(t)\sigma(t)^\top)^{-1}\sigma(t)\left[1-{\gamma}(t,T)e^{-\int_t^T r(s) ds}\right]X(t)dW(t),\;\;t\in[0,T], \\
X(0)=x_0.
\end{cases}
\end{equation}
Moreover, the following is the na\"ive policy:
\begin{equation}\label{PI}\bm{\pi}^*(t,x)=-[\sigma(t)\sigma(t)^\top]^{-1}B(t)^\top [1-\gamma(t,T)e^{-\int_t^T r(s) ds}]x,\;\;(t,x)\in [0,T]\times \R.
\end{equation}
\end{theorem}

A proof of Theorem \ref{theorem1} is delayed to Appendices \ref{proof_theorem1}.

Note that the explicitly presented policy (\ref{PI}) indeed does not depend on any initial pair $(s,y)$ and, in particular, on $(0,x_0)$. This means that even if the wealth process of the $2^{-n}$-committed is constructed from an arbitrarily different initial pair $(s,y)$, it will lead to the same na\"ive policy (\ref{PI}). On the other hand, it generates $X$ as its wealth process for the given initial $(0,x_0)$.

\section{Comparison between Na\"ive and Other Types of Policies}\label{Naive_comparisons}
In the continuous-time MV literature, two types of equilibrium policies by consistent planners have been introduced and  studied:  the {\it weak} equilibrium policies  by \cite{BMZ14} and the {\it regular} equilibrium policies by \cite{HJ17}.  In this section, we compare the  na\"ive policies with these two types of equilibrium policies as well as the pre-committed ones, in a Black--Scholes market.

\subsection{Weak and regular equilibrium policies}
 We first review the two types of equilibrium strategies, whose definitions can be found, in slight variants of the MV formulation, in \cite{BMZ14} and \cite{HJ17} respectively.

 Given $(s,y)\in [0,T]\times\R$, \cite{BMZ14} consider the following problem:

 \begin{equation}\label{BjorkP}
 \max\limits_{\pi(\cdot)\in {\cal U}_{s,y}}\ J(s,y;\pi(\cdot)):=\E_{s,y}[X(T)]-\frac{\alpha(s,y)}{2}\var_{s,y}(X(T))
 \end{equation}
\begin{equation}\label{ft0BjorkP}
\text{subject to}\
(X(\cdot),\pi(\cdot))\ \text{satisfy}\ (\ref{wealtheq}) \mbox{ with } X(s)=y.
\end{equation}
In the objective function of this problem, there is a risk-aversion term $\alpha(s,y)>0$ that depends on the initial time $s$ and initial state $y$; see \cite{BMZ14} for the many discussions on the motivation of such a varying risk-aversion term.\footnote{\cite{BMZ14} consider only the state-dependent risk aversion
$\alpha(s,y)=\alpha(y)$, but the method and results therein readily extend to the time--state dependent case presetned here.}  The problem is again time-inconsistent.
\cite{BMZ14} study the behavior of a consistent planner by considering the equilibrium policies defined as follows. 
Given an admissible (time-consistent) policy  $\hat{\bm{\pi}}(\cdot,\cdot)$, construct a new policy  $\bm{\pi}_h$ by
\begin{equation}\label{pihcons}
    \bm{\pi}_h(t,x): =
    \begin{cases}
    \pi, &t\in[s,s + h),\;x\in\R,\\
    \hat{\bm{\pi}}(t,x), &t\in[0,T]\setminus [s,s + h),\;x\in\R,
    \end{cases}
\end{equation}
where $\pi\in\mathbb{R}^m$, $h > 0$ and $s\in[0,T)$ are aribitrarily given.
Let $\pi(\cdot)$ and $\pi_h(\cdot)$ be respectively the portfolio processes generated by $\hat{\bm{\pi}}$ and $\bm{\pi}_h$ starting from $(s,y)$.
We say that $\hat{\bm{\pi}}$ is a {\it weak equilibrium policy} if the  pertubed policy $\bm{\pi}_h$ is admissible and
\begin{equation}
    \lim\limits_{h\to 0} \inf \frac{J(s,y;\hat{\pi}) - J(s,y;\pi_h)}{h} \geq 0,
\end{equation}
for all $\pi\in\mathbb{R}^m$ and $(s,y)\in [0,T)\times\R$.

On the other hand, \cite{HJ17} formulate the following problem:
\begin{equation}\label{XueP}
\min\limits_{\pi(\cdot)\in {\cal U}_{s,y}}\ \var_{s,y}(X(T))
\end{equation}
\begin{equation}\label{ft}
\text{subject to}\
\begin{cases}
\E_{s,y}[X(T)]=  L(s,y),\\
(X(\cdot),\pi(\cdot))\ \text{satisfy}\ (\ref{wealtheq}) \mbox{ with } X(s)=y,
\end{cases}
\end{equation}
where $L(s,y)$ indicates the expected terminal wealth target when the initial pair is $(s,y)$.\footnote{In the original formulation of \cite{HJ17}, the expected terminal wealth constraint is $\E_{s,y}[X(T)]\geq  L(s,y)$,
which is equivalent to the equality constraint formulated here.} When $L(s,y) = yf(s,T)$, the problem (\ref{XueP})--(\ref{ft}) reduces to the problem (\ref{prob1})--(\ref{ft0}). \cite{HJ17} also study a consistent planner, except that they use the notion of  {\it regular} equilibrium policies which is very different from that of the weak equilibrium policies. Specifically, an admissible, time-consistent policy  $\hat{\bm{\pi}}$ is called a {\it regular equilibrium policy} if for any $(s,y)\in [0,T)\times\R$, any $\pi\in\mathbb{R}^m$ such that $\bm{\pi}_h$ constructed by (\ref{pihcons}) is admissible for sufficeintly small $h>0$, we have\footnote{Here, the term ``admissible" requires the corresponding portfolio processes generated by the relevant policies for $(s,y)$ to also satisfy the expectation constraint in (\ref{ft}).}
%
%
%
\begin{equation}
\var_{s,y}(X^{\pi_h}(T))-\var_{s,y}(X^{\hat{\pi}}(T))\geq0
\end{equation}
for sufficiently small $h>0$, where $X^{\hat{\pi}}(T)$ and $X^{\pi_h}(T)$ are the terminal wealth values, both starting from $(s,y)$ and  under $\hat{\bm{\pi}}$ and $\bm{\pi}_h$ respectively.

The difference between the problems (\ref{BjorkP})--(\ref{ft0BjorkP}) and (\ref{XueP})--(\ref{ft}) is that the former uses a weighting coefficient $\alpha(s,y)/2$ in its objective function while the latter takes $L(s,y)$ in its constraint. The two problems are related via the Lagrange multiplier method. As a result, if we choose $\alpha(s,y)$ and $L(s,y)$ in a certain way, then the respective pre-committed optimal polices  for the two problems coincide, as stipulated in the following proposition.
\begin{proposition}\label{BjorkXueEq}
If
\begin{equation}\label{GL}
\frac{1}{\alpha(s,y)}e^{\int_s^T\rho(t)dt}+ye^{\int_s^Tr(t)dt}
=\frac{L(s,y)-e^{\int_s^T[r(t)-\rho(t)]dt}y}{1-e^{-\int_s^T\rho(t)dt}},\;\;\forall (s,y)\in [0,T]\times\R
\end{equation}
holds, then the pre-committed optimal policies  for (\ref{BjorkP})--(\ref{ft0BjorkP}) and (\ref{XueP})--(\ref{ft}) are the same for any $(s,y)\in [0,T]\times\R$.
\end{proposition}
\begin{proof}
It follows from the equations (5.12), (5.1) and (6.7) in \cite{ZL00}
that the pre-committed optimal policy  of (\ref{BjorkP})--(\ref{ft0BjorkP}) is
\begin{equation}
\begin{aligned}
\bar{\bm{\pi}}_*(t,x;s,y)=-[\sigma(t)\sigma(t)^\top]^{-1}B(t)^\top \left[x-{\bar\gamma}(s,T,y) e^{-\int_t^T r(v) dv}\right],\;\;(t,x)\in [s,T]\times \R,
\end{aligned}
\end{equation}
where
$$ \bar{\gamma}(s,T,y):= \frac{1}{\alpha(s,y)}e^{\int_s^T \rho(v)dv }+e^{\int_s^T r(v) dv}y.$$

On the other hand, it follows from \cite[Theorem 2.1]{LZ06} that  the precommitted strategy of (\ref{XueP})--(\ref{ft}) is
\begin{equation}
\begin{aligned}
\tilde{\bm{\pi}}_*(t,x;s,y)=-[\sigma(t)\sigma(t)^\top]^{-1}B(t)^\top \left[x-{\tilde\gamma}(s,T,y) e^{-\int_t^T r(v) dv}\right],\;\;(t,x)\in [s,T]\times \R,
\end{aligned}
\end{equation}
where
$$\tilde{\gamma}(s,T,y):=\frac{L(s,y) - e^{\int_s^T [r(v) - \rho(v)] dv}y}{1 - e^{-\int_s^T\rho(v) dv}}.$$

It is now evident that if (\ref{GL}) is satisfied, then  $\bar{\gamma}(s,T,y) \equiv  \tilde{\gamma}(s,T,y)$ leading to $\bar{\bm{\pi}}_*(t,x;s,y) \equiv \tilde{\bm{\pi}}_*(t,x;s,y)$.
\end{proof}

The condition (\ref{GL}) ensures that the pre-committed solutions of the two problems
coincide. As a result, the na\"ive policies of the two problems are also identical because
they are obtained via the limit of pre-committed policies. However,  (\ref{GL}) does not necessarily lead to the same weak/regular equilibrium policies  of the two problems, because equilibrium policies are not based on pre-committed ones.


\subsection{Comparisons}

We now compare the na\"ive policies with the weak/regular equilibrium policies and the pre-committed polices, in a Black--Scholes market for simplicity. Specifically, there is a risk-free asset and only one risky asset (i.e. $m=1$) with $r(t)\equiv r>0$,
$B(t)\equiv b-r>0$, $\sigma(t)\equiv \sigma>0$. As a result, $\rho(t)\equiv \rho=(\frac{b-r}{\sigma})^2>0$.

 We carry out the comparison for two cases. In Subsection \ref{gammax}, we choose $\alpha(s,y)=\frac{\alpha}{y}$ for some constant $\alpha>0$ in the problem (\ref{BjorkP})--(\ref{ft0BjorkP}),  which is also a case examined closely in  \cite{BMZ14}. Subsection \ref{Ltx} studies the case when  $L(s,y)=ye^{k(T-s)}$ for some constant $k>r$ in the problem (\ref{XueP})--(\ref{ft}). In each case, we choose $f(\cdot,\cdot)$, $L(\cdot,\cdot)$ and $\alpha(\cdot,\cdot)$ in such a way (e.g. to  satisfy (\ref{GL})) that the different formulations of the MV problem are consistent in their respective pre-committed optimal policies.

\subsubsection{The case $\alpha(s,y)=\frac{\alpha}{y}$}\label{gammax}
When $\alpha(s,y)=\frac{\alpha}{y}$, the corresponding  $L$ according to (\ref{GL}) is
\begin{equation}\label{Le1}
L(s,y)=y\left[\frac{1}{\alpha}\left(e^{(T-s)\rho}-1+\alpha e^{(T-s)r}\right)\right],
\end{equation}
whereas the corresponding $f$ is
\begin{equation}\label{f1}
f(s,T)=\frac{1}{\alpha}\left[e^{(T-s)\rho}-1+\alpha e^{(T-s)r}\right].
\end{equation}
It is easy to check that this $f$ satisfies Assumption \textbf{(A3)}. By Theorem \ref{theorem1},
the na\"ive policy is
\begin{equation}\label{PII}
\bm{\pi}^*(t,x)=-\frac{b-r}{\sigma^2}\left[1-\frac{f(t,T)
-e^{(r-\rho)(T-t)}}{1-e^{-\rho(T-t)}}
e^{-r(T-t)}\right]x,\;\;(t,x)\in [0,T]\times \R.
\end{equation}
Substituting the expression of $f$ in (\ref{f1}) into the above and going through some simple computation, we finally get
\begin{equation}\label{PIII}
\bm{\pi}^*(t,x)=\frac{b-r}{\alpha\sigma^2}e^{(\rho-r)(T-t)}x,\;\;(t,x)\in [0,T]\times \R.
\end{equation}
The {\it risky weight} function of this policy, defined as the ratio between the dollar amount in the stock and the total wealth and denoted by $c_{na}$, is thereby
\begin{equation}\label{cnacase1}
c_{na}(t):=\frac{\bm{\pi}^*(t,x)}{x}=\frac{b-r}{\alpha\sigma^2}e^{(\rho-r)(T-t)},
\;\;t\in [0,T],
\end{equation}
which turns out to be a function of $t$ only.

On the other hand, when  $\alpha(s,y)=\frac{\alpha}{y}$, Theorem 4.6 in \cite{BMZ14}
gives
the weak equilibrium policy  of the problem (\ref{BjorkP})--(\ref{ft0BjorkP}) as
\begin{equation}
\bm{\pi}_{we}(t,x)=c_{we}(t)x,
\end{equation}
where $c(t)\equiv c_{we}(t)$ is the unique solution to  the following integral equation
\begin{equation}\label{cwe}
c(t)=\frac{b-r}{\alpha\sigma^2}\left[e^{-\int_t^T[r+(b-r)c(s)+\sigma^2c(s)^2]ds}+\alpha e^{-\int_t^T \sigma^2 c(s)^2 ds}-\alpha\right].
\end{equation}
Similarly, $c_{we}$ is the risky weight function of the weak equilibrium policy.

Finally,  we can rewrite (\ref{Le1}) as
\begin{equation}\label{Le11}
L(s,y)=ye^{\int_s^T[r+\psi(t)]dt}
\end{equation}
where
\begin{equation}\label{psi}
\psi(t):=\frac{r+(\rho-r)e^{\rho(T-t)}}{\alpha e^{(T-t)r}+e^{\rho(T-t)}-1}.
\end{equation}

Applying Theorem 1-i in \cite{HJ17} and noting that the solution to the problem (2.10) therein is $v^*(t)=\frac{\psi(t)}{b-r}$, we obtain the regular equilibrium policy for (\ref{XueP})--(\ref{ft}) to be
\begin{equation}
\bm{\pi}_{re}(t,x)=c_{re}(t)x,
\end{equation}
where
\begin{equation}
c_{re}(t):=\frac{\psi(t)}{b-r}=\frac{1}{b-r}\frac{r+(\rho-r)e^{\rho(T-t)}}{\alpha e^{(T-t)r}+e^{\rho(T-t)}-1},\;\;t\in [0,T]
\end{equation}
is the risky weight of this  equilibrium policy at $t\in[0,T]$.

The following proposition shows that  the na\"ive policy  allocates {\it strictly} more weight to the risky asset than the two equilibrium policies at {\it any} time before $T$.
\begin{proposition}\label{compare_naive_1}
In the Black--Scholes market, if $\alpha(s,y)=\frac{\alpha}{y}$, then we have
	$$c_{we}(t)<c_{na}(t),\;\;c_{re}(t)<c_{na}(t),\;\;\forall t\in[0,T),$$
for any $\alpha>0$.
\end{proposition}

\begin{proof}
Let us first prove $c(t)\equiv c_{we}(t)<c_{na}(t)\;\;\forall t\in[0,T)$.
We have the obvious inequality
\begin{equation}\label{oi}
\rho+(b-r)c(s)+\sigma^2c(s)^2>0,\;\;\forall s\in[0,T)
\end{equation}
because $\Delta:=(b-r)^2-4\rho\sigma^2=-3(b-r)^2<0$. Recalling that $c(\cdot)$ satisfies (\ref{cwe}), we deduce
\begin{equation}
\begin{aligned}
c_{we}(t)&=\frac{b-r}{\alpha\sigma^2}\left[e^{-\int_t^T[r+(b-r)c(s)+\sigma^2c(s)^2]ds}+\alpha e^{-\int_t^T \sigma^2 c(s)^2 ds}-\alpha\right]\\
&\leq \frac{b-r}{\alpha\sigma^2}e^{-\int_t^T[r+(b-r)c(s)+\sigma^2c(s)^2]ds}\\
&< \frac{b-r}{\alpha\sigma^2}e^{-\int_t^T(r-\rho)ds}\\
&=\frac{b-r}{\alpha\sigma^2}e^{(\rho-r)(T-t)}=c_{na}(t),\ \forall t\in[0,T).
\end{aligned}
\end{equation}

Next, we prove $c_{re}(t)<c_{na}(t)\;\;\forall t\in[0,T)$. Indeed
\[ \begin{aligned}
c_{re}(t)&=\frac{1}{b-r}\frac{r+(\rho-r)e^{\rho(T-t)}}{\alpha e^{(T-t)r}+e^{\rho(T-t)}-1}\\
&<\frac{1}{b-r}\frac{\rho e^{\rho(T-t)}}{\alpha e^{(T-t)r}}\\
&=\frac{b-r}{\alpha\sigma^2}e^{(\rho-r)(T-t)}=c_{na}(t),\ \forall t\in[0,T).
\end{aligned}
\]

The proof is complete.

\end{proof}

So  na\"ive policies take more risk exposure than the two types of equilibrium policies. It is interesting to compare the na\"ivet\`e also with a pre-committer, realizing that the former strives to follow the latter at {\it every} initial pair $(s,y)$. Take $(s,y)=(0,x_0)$ for example. The pre-committer's expected terminal wealth is
\begin{equation}\label{f0T} \E_{0,x_0}[X_*(T)]=x_0f(0,T)={x_0}e^{rT}\frac{1}{\alpha}\left[e^{(\rho -r) T}-e^{-rT}+\alpha \right],
\end{equation}
noting (\ref{f1}). Although the na\"ivet\`e's original expected target return was also $x_0f(0,T)$ at $(0,x_0)$, he changes mind all the time subsequently so his {\it actual} target return at $(0,x_0)$ can be significantly deviate from the original one. To see this, plugging in the na\"ive policy (\ref{PIII}) to the wealth equation (\ref{wealtheq}) to obtain
{\small
\begin{equation}\label{wealtheqex}
dX^*(t)=\left[rX^*(t)+\frac{1}{\alpha}\rho e^{(\rho-r)(T-t)}X^*(t)\right]dt+\frac{b-r}{\alpha\sigma}e^{(\rho-r)(T-t)}X^*(t)dW(t),\ t\in[0,T];\;\; X^*(0)=x_0.
\end{equation}}
Taking the integral form of this SDE and applying expectation on both sides, we get an ODE in terms of $\E[X^*(\cdot)]\equiv \E_{0,x_0}[X^*(\cdot)]$. Solving this ODE we arrive at
\begin{equation}\label{f0T0} \E_{0,x_0}[X^*(T)]={x_0}e^{rT}e^{\frac{1}{\alpha}\frac{\rho}{\rho-r}[e^{(\rho -r)T}-1]}.
\end{equation}
Recall that $\alpha>0$ is the risk aversion coefficient,  and the smaller $\alpha$ the less risk averse the agent is. Comparing (\ref{f0T0}) with  (\ref{f0T}) and noting that $\frac{\rho}{\rho-r}[e^{(\rho -r)T}-1]>0$ always holds, the na\"ivet\`e's expected terminal wealth is larger than the pre-committer's when $\alpha$ is small, and the former grows exponentially fast while the latter does only linearly in $\alpha^{-1}$ as $\alpha\to 0$. So a na\"ive policy ends up achieving   a much higher   expected terminal wealth than a pre-committed one which is also his {\it originally} planned target.\footnote{This also reconciles with the previously proved fact that na\"ive policies are more exposed to the stock than equilibrium ones.} However, this by no means implies that the former is superior to the latter because in an MV model there are two criteria and the variance is as important as the return. Indeed, it is straightforward to check that the na\"ive policy (\ref{PII}) is {\it different} from the {\it unique} pre-committed optimal policy (\ref{us}) under the {\it new} expected terminal wealth (\ref{f0T0}), hence must be MV {\it inefficient}.\footnote{Alternatively, one can  calculate $\var_{0,x_0}(X^*(T))$ and show that it is strictly larger than the right hand side of (\ref{ef}) with the expected terminal wealth given by (\ref{f0T0}) and  $(s,y)=(0,x_0)$. In other words,
$(\E_{0,x_0}[X^*(T)],\var_{0,x_0}(X^*(T)))$ lies {\it off} the efficient frontier (\ref{ef}).
Details are left to interested readers.} In other words, the na\"ive policy (\ref{PII}) takes more risk than it needs to - as dictated by the efficient frontier --  in order to achieve a higher expected terminal wealth (\ref{f0T0}).

To sum, in the current MV setting, a na\"ive policy is more risk-loving than the other types of polices while expecting higher terminal wealth. Although at every $(s,y)$ it tries to follow the pre-committed optimal policy, the actual policy turns out to be very different. It is MV inefficient and certainly not ``dynamically optimal" in any sense at any given $(s,y)$.

\subsubsection{The case $L(t,x)=xe^{k(T-t)}$}\label{Ltx}
We now consider the case when $L(t,x)=xe^{k(T-t)}$, where  $k>r$ (otherwise the problem (\ref{XueP})--(\ref{ft}) is trivial). The corresponding $f$ is
$f(t,T)=e^{k(T-t)},$
which satisfies Assumption \textbf{(A3)}.
Substituting this into (\ref{PII}) we obtain the na\"ive policy
\[
\bm{\pi}^*(t,x)=c_{na}(t)x,\;\;(t,x)\in [0,T]\times \R
\]
where the risky weight is
\begin{equation}\label{cnacase2}
c_{na}(t)=\frac{\bm{\pi}^*(t,x)}{x}
=\frac{b-r}{\sigma^2}\frac{e^{(k-r)(T-t)}-1}{1-e^{-\rho(T-t)}},
\;\;t\in [0,T].
\end{equation}

Next, it follows from  (\ref{GL}) that the corresponding
\begin{equation}\label{ge1}
\alpha(s,y)=\frac{\phi(s)}{y}
\end{equation}
where $\phi(s):=\frac{e^{\rho(T-s)}-1}{e^{k(T-s)}-e^{r(T-s)}}>0$.
Again, by Theorem 4.6 in \cite{BMZ14} we get the weak equilibrium policy  of the problem (\ref{BjorkP})--(\ref{ft0BjorkP}) to be
\[
\bm{\pi}_{we}(t,x)=c_{we}(t)x,
\]
where $c(t)\equiv c_{we}(t)$ uniquely solves
\begin{equation}\label{cwe2}
c(t)=\frac{b-r}{\phi(t)\sigma^2}\left[e^{-\int_t^T[r+(b-r)c(s)+\sigma^2c(s)^2]ds}+\phi(t) e^{-\int_t^T \sigma^2 c(s)^2 ds}-\phi(t)\right].
\end{equation}

Finally, by Theorem 1-i in \cite{HJ17}, the regular equilibrium policy for (\ref{XueP})--(\ref{ft}) is
\[
\bm{\pi}_{re}(t,x)=c_{re}(t)x,
\]
where
\begin{equation}
c_{re}(t):=\frac{k-r}{b-r},\;\;t\in [0,T].
\end{equation}

\begin{proposition}\label{compare_naive_2}
In the Black-Scholes market, if $L(t,x)=xe^{k(T-t)}$, then we have
	$$c_{we}(t)<c_{na}(t),\;\;c_{re}(t)<c_{na}(t),\;\;\forall t\in[0,T),$$
for any $k>r$.
\end{proposition}
\begin{proof}
It follows from (\ref{cwe2}) that
\begin{equation}
\begin{aligned}
c_{we}(t)\equiv c(t)&=\frac{b-r}{\phi(t)\sigma^2}\left[e^{-\int_t^T[r+(b-r)c(s)+\sigma^2c(s)^2]ds}+\phi(t) e^{-\int_t^T \sigma^2 c(s)^2 ds}-\phi(t)\right]\\
&\leq \frac{b-r}{\phi(t)\sigma^2}e^{-\int_t^T[r+(b-r)c(s)+\sigma^2c(s)^2]ds}\\
&< \frac{b-r}{\phi(t)\sigma^2}e^{-\int_t^T(r-\rho)ds}\\
&=\frac{b-r}{\phi(t)\sigma^2}e^{(\rho-r)(T-t)}\\
&=\frac{b-r}{\sigma^2}\frac{e^{(k-r)(T-t)}-1}{1-e^{-\rho(T-t)}}=c_{na}(t),\ \forall t\in[0,T),
\end{aligned}
\end{equation}
where we have utilized  (\ref{oi}) to get the second inquality and noted the definition of $\phi(\cdot)$ to obtain the second to the last equality.

Next, applying the general inequality
\[ \frac{e^x-1}{1-e^{-y}}>\frac{x}{y},\;\;\forall x>0,\;y>0,
\]
we deduce
\[c_{na}(t)
=\frac{b-r}{\sigma^2}\frac{e^{(k-r)(T-t)}-1}{1-e^{-\rho(T-t)}}
>\frac{b-r}{\sigma^2}\frac{k-r}{\rho}=\frac{k-r}{b-r}=c_{re}(t).
\]
The proof is complete.
\end{proof}

We can also show that the na\"ive policy is not MV efficient with respect to any initial $(s,y)$ in the current case. Because the analysis is similar to  that in the previous subsection, we omit the details here.

\section{Conclusions}\label{Naive_conclusions}

In this paper we define precisely and derive rigorously the policies implemented by a na\"ive agent, a notion originally put forth by \cite{SR56}, for a continuous-time Markowitz model that is intrinsically time inconsistent. Such an agent attempts to optimize at any given time but, since optimal policies depend on when and where one makes them in a time-inconsistent problem, in effect constantly changes his policies. Ironically, the policy a na\"ivet\'e actually executes may be anything but he originally desired. At any given time and state he sets an expected investment target and wants to achieve mean--variance efficiency but we show that his final policy ends up with a (much) higher target return and an even higher variance that overall becomes mean--variance {\it in}efficient. Moreover, na\"ive policies are universally riskier than their consistent planning counterparts.

Studying na\"ive behaviors in continuous-time problems is a nearly uncharted research area. From a behavioral economics perspective, it is fascinating to inquire and understand how an originally well-intended policy may go wrong or even go opposite when one insists on optimizing {\it all the time}. The definition of na\"ive policies and the approach to derive them in this paper is generalizable to other types of problems such as those with non-exponential discounting and probability weighting.
As such, we hope the paper has also set a stage for further study of these problems.


\newpage

\noindent\textbf{{\Large Appendices}}
\appendix

\section{Proof of Proposition \ref{normfinite}}\label{proof_normfinite}
    The main idea of the proof is to find a \emph{deterministic} function  $Y$ to bound $X_n^2$, which is stated in the following lemma.

\begin{lemma}\label{boundode}
Let $Y$ satisfying the following ODE
\begin{equation}
dY(s)=\left[R^*+(\gamma^{*})^2e^{-2\int_s^T r(v) dv }\rho(s)\right]Y(s)ds,\;s\in[0,T];\;\;
Y(0)=x_0,
\end{equation}
where
$$R^*:=\max\limits_{0\leq s\leq T}|2r(s)-\rho(s)|,\ \gamma^{*}:=\max\limits_{0\leq s\leq T}\gamma(s,T).$$
Then, we have, for every $k=0,1,...,2^n-1$,
$$\E[X_*(s;t_k)^2]\leq Y(s),\ \;s\in[t_k,t_{k+1}].$$
\end{lemma}
\begin{proof}
By Assumptions {\bf (A1)--(A3)},  it is clear that $R^*<\infty$ and $\gamma^*<\infty$.

Recall $X_*(\cdot;t_k)$ satisfies the SDE (\ref{rec}) on $[t_k,t_{k+1}]$ for $k=0,1,...,2^n-1$. Applying It\^o's formula to $X_*(t;t_k)^2$ and then taking
conditional expectation on $\mathcal{F}_{t_k}$ we obtain
the ($\omega$-wise) ODE
{\small
\begin{equation}\label{ode0}
\begin{cases}
d\E[X_*(t;t_k)^2|\mathcal{F}_{t_k}]=\left\{(2r(t)-\rho(t))\E[X_*(t;t_k)^2|\mathcal{F}_{t_k}]
+\gamma(t_k,T)^2\rho(t)e^{-2\int_t^T r(v) dv }X_*(t_k;t_k)^2\right\}dt,\;\; t\in[t_k,t_{k+1}],\\
\E[X_*(t_k;t_k)^2|\mathcal{F}_{t_k}]=X_*(t_k;t_k)^2.
\end{cases}
\end{equation}}
Consider a new stochastic process $Z(\cdot;t_k)$ which satisfies the ODE on $[t_k,t_{k+1}]$ for $k=0,1,...,2^n-1$:
\begin{equation}\label{ode1}
\begin{cases}
dZ(t;t_k)=\left[R^*Z(t;t_k)+\gamma(t_k,T)^2\rho(t)e^{-2\int_t^T r(v) dv }X_*(t_k;t_k)^2\right]dt,\;\; t\in[t_k,t_{k+1}],\\
Z(t_k;t_k)=X_*(t_k;t_k)^2.
\end{cases}
\end{equation}
Because  $|2r(t)-\rho(t)|\leq R^*$, $t\in[0,T]$, a comparison theorem of ODEs yields
 \begin{equation}\label{comp0}
\E[X_*(t;t_k)^2|\mathcal{F}_{t_k}]\leq Z(t;t_k),\;\mbox{a.s.}, \;\;k=0,1,...,2^n-1.
\end{equation}

Now, we construct another  stochastic process $\bar Z(\cdot;t_k)$ on $[t_k,t_{k+1}]$ for $k=0,1,...,2^n-1$:
\begin{equation}\label{eq1}
\begin{cases}
d\bar Z(t;t_k)=\left[R^*+(\gamma^*)^2\rho(t)e^{-2\int_t^T r(v) dv }\right]\bar Z(t;t_k)dt,\;\; t\in[t_k,t_{k+1}],\\
\bar Z(t_k;t_k)=X_*(t_k;t_k)^2.
\end{cases}
\end{equation}

It follows from (\ref{ode1}) that $Z(t;t_k)$ increases in $t\in[t_k,t_{k+1}]$; hence $Z(t;t_k)\geq X_*(t_k;t_k)^2$ for $t\in[t_k,t_{k+1}]$. Then, we get
\begin{equation}\label{eq2}
\begin{aligned}
\frac{dZ(t;t_k)}{dt}&=R^*Z(t;t_k)+\gamma(t_k,T)^2\rho(t)e^{-2\int_t^T r(v) dv }X_*(t_k;t_k)^2\\
&\leq \left[R^*+\gamma(t_k,T)^2\rho(t)e^{-2\int_t^T r(v) dv }\right]Z(t;t_k)\\
&\leq \left[R^*+(\gamma^*)^2\rho(t)e^{-2\int_t^T r(v) dv }\right]Z(t;t_k).
\end{aligned}
\end{equation}
Comparing (\ref{eq1}) and (\ref{eq2}), we conclude  from the Grownwall inequality that 
\begin{equation}\label{comp1}
Z(t;t_k)\leq \bar Z(t;t_k), \;\mbox{a.s.},\; t\in[t_k,t_{k+1}],\;k=0,1,...,2^n-1.
\end{equation}

To finish the proof we use mathematical induction on $k$. When $k=0$, $t\in [0,t_1]$, it follows from (\ref{comp0}) and (\ref{comp1}) that
\begin{equation}
\E[X_*(t;0)^2]=\E[\E[X_*(t;0)^2|\mathcal{F}_0]]\leq \E[Z(t;0)]\leq \E[\bar Z(t;0)]=Y(t).
\end{equation}
Now, assume that when $k=m-1$, the following holds:
\begin{equation}\label{comp3}
\E[X_*(t;t_{m-1})^2]\leq Y(t),\;\; t\in [t_{m-1},t_{m}].
\end{equation}
By (\ref{comp0}) and (\ref{comp1}) we obtain
\begin{equation}\label{comp5}
\begin{aligned}
\E[X_*(t;t_m)^2]&=\E[\E[X_*(t;t_m)^2|\mathcal{F}_{t_m}]]\\
&\leq \E[Z(t;t_m)]\\
&\leq \E[\bar Z(t;t_m)],\;\;t\in [t_{m},t_{m+1}]
\end{aligned}
\end{equation}
where the initial value of $\E[\bar Z(\cdot;t_m)]$ on $[t_{m},t_{m+1}]$ is $\E[X_*(t_m;t_{m})^2]\equiv \E[X_*(t_m;t_{m-1})^2]$. However,  (\ref{comp3}) gives $\E[X_*(t_m;t_{m-1})^2]\leq Y(t_m)$, whereas  $\E[\bar Z(\cdot;t_m)]$ and $Y(\cdot)$ satisfy the same ODE on $[t_m,t_{m+1}]$. Thus $\E[\bar Z(t;t_m)]\leq Y(t)$ on $[t_m,t_{m+1}]$. Combining with (\ref{comp5}), we get the desired result.  
\end{proof}

We are now ready to  prove Proposition \ref{normfinite}.
By Lemma \ref{boundode}, we have
\begin{equation}\label{bound1}
\begin{aligned}
||X_n||^2&=\E\int_0^T X_n(s)^2 ds
=\sum\limits_{k=1}^{2^n}\int_{t_{k-1}}^{t_k}\E[X_*(s;t_{k-1})^2 ]ds\\
&\leq \sum\limits_{k=1}^{2^n}\int_{t_{k-1}}^{t_k}Y(s)ds=\int_0^T Y(s)ds<\infty.
\end{aligned}
\end{equation}

\section{Proof of Theorem \ref{theorem1}}\label{proof_theorem1}

 To ease notation  we use the following
\begin{equation}\label{nota1}
\begin{cases}
\gamma(t):=\gamma(t,T),\;A(t):=r(t)-\rho(t), \;
C(t):=e^{-\int_t^T r(v) dv}\rho(t), \\
D(t):=B(t)(\sigma(t)\sigma(t) ^\top )^{-1}\sigma(t)e^{-\int_t^T r(v) dv}, \;
F(t):=B(t)(\sigma(t)\sigma(t)^\top )^{-1}\sigma(t),
\end{cases}
\end{equation}
with which we rewrite the SDE (\ref{rec})  as
\begin{equation}\label{nota2}
\begin{cases}
dX_*(t;t_k)=\left[A(t)X_*(t;t_k)+\gamma(t_k)C(t)X_*(t_k;t_{k-1})\right]dt\\
\ \ \ +\left[-F(t)X_*(t;t_k)+\gamma(t_k)D(t)X_*(t_k;t_{k-1})\right]dW(t),\;\; t\in[t_k,t_{k+1}], \\
X_*(t_k;t_k)=X_*(t_k;t_{k-1}).
\end{cases}
\end{equation}

Denote
\[ A^*:=\max\limits_{t\in [0,T]}|A(t)|^2,\;C^*:=\max\limits_{t\in [0,T]}|C(t)|^2,\;
D^*:= \max\limits_{t\in [0,T]}||D(t)||^2,\;F^*:= \max\limits_{t\in [0,T]}||F(t)||^2,\]
which are all finite due to the boundedness assumptions in \textbf{(A1)} and \textbf{(A2)}.

In order to prove Theorem \ref{theorem1}, we need the following lemma.
\begin{lemma}\label{lemmaclose}
The process  $X_n$ defined by (\ref{pasting}) satisfies
$$\lim\limits_{n\to\infty}\max\limits_{k\in\{0,...,2^n-1\},s\in[t_k,t_{k+1}]}
\E|X_n(s)-X_n(t_k)|^2= 0.$$
\end{lemma}

\begin{proof}
For $s\in[t_k,t_{k+1}]$, we bound the term $\E|X_n(s)-X_n(t_k)|^2$ as follows:
\begin{equation}
\begin{aligned}
&\E|X_n(s)-X_n(t_k)|^2=\E|X_*(s;t_k)-X_*(t_k,t_{k-1})|^2\\
&\leq 2\E\left[\int_{t_k}^s \left(A(t)X_*(t;t_k)+\gamma(t_k)C(t)X_*(t_k;t_{k-1})\right)dt\right]^2 \\
&+2\E\left[\int_{t_k}^s \left(-F(t)X_*(t;t_k)+\gamma(t_k)D(t)X_*(t_k;t_{k-1})\right)dW(t)\right]^2.
\end{aligned}
\end{equation}
For bounding the first term on the right side of the above,  we have by the Cauchy--Schwartz inequality
\begin{equation}\label{diffsquare}
\begin{aligned}
&\E\left[\int_{t_k}^s \left(A(t)X_*(t;t_k)+\gamma(t_k)C(t)X_*(t_k;t_{k-1})\right)dt\right]^2 \\
&\leq (s-t_k)\int_{t_k}^s \E\left|A(t)X_*(t;t_k) + \gamma(t_k)C(t)X_*(t_k;t_{k-1})\right|^2dt\\
&\leq (s-t_k)\int_{t_k}^s 2\E|A(t)X_*(t;t_k)|^2 + 2\E|\gamma(t_k)C(t)X_*(t_k;t_{k-1})|^2dt\\
&\leq (s-t_k)\int_{t_k}^s \left(2A^*\E|X_*(t;t_k)|^2 + 2\gamma^*C^*\E|X_*(t_k;t_{k-1})|^2\right)dt\\
&\leq (s-t_k)\int_{t_k}^s (2A^*+ 2\gamma^*C^*)Y(T)dt = (2A^*+ 2\gamma^*C^*)(s-t_k)^2Y(T),
\end{aligned}
\end{equation}
where  the last inequality follows from Lemma \ref{boundode} and the fact that $Y(s)$ is increasing in $s\in[0,T]$.

For the second term, by virtue of It\^o's isometry, we similarly have
\begin{equation}
\E\left[\int_{t_k}^s \left(-F(t)X_*(t;t_k)+\gamma(t_k)D(t)X_*(t_k;t_{k-1})\right)dW(t)\right]^2\leq (2\gamma^*D^*+ 2F^*)(s-t_k)Y(T).
\end{equation}
Combining the above, we obtain
\begin{equation}
\begin{aligned}
\E|X_n(s)-X_n(t_k)|^2&\leq 4(s-t_k)(A^*+\gamma^*C^* + \gamma^*D^* + F^*)Y(T), \;s\in[t_k,t_{k+1}].
\end{aligned}
\end{equation}
Thus,
$$\max\limits_{k\in\{0,...,2^n-1\},s\in[t_k,t_{k+1}]}\E[X_n(s)-X_n(t_k)]^2\leq \frac{4T}{2^n}(A^*+\gamma^*C^* + \gamma^*D^* + F^*)Y(T)\to 0$$
as $n\to\infty$.
\end{proof}

Because $X_n \to X\ \text{weakly in $L^2_{\mathcal{F}}([0,T];\mathbb{R})$}$, it follows from
Mazur's lemma that for each integer $n\geq 1$, there exists a positive integer $N(n)$ and a convex combination $V_n:=\sum_{k=n}^{N(n)} a_{k,n} X_k$, where $a_{k,n}\geq0$ and $\sum_{k=n}^{N(n)} a_{k,n}=1$, such that
\begin{equation}
V_n \to X\ \text{strongly in $L^2_{\mathcal{F}}([0,T];\mathbb{R})$}.
\end{equation}
By the definition of $V_n$, it satisfies the SDE
\begin{equation}
\begin{cases}
dV_n(t)=[A(t)V_n(t)+C(t)U_n(t)]dt+[-F(t)V_n(t)+D(t)U_n(t)]dW(t),\\
V_n(0)=x_0,
\end{cases}
\end{equation}
where
$$U_n(t):=\sum_{k=n}^{N(n)} a_{k,n}[\gamma(m_{t,k})X_k(m_{t,k})],\ m_{t,k}:=\frac{N}{2^k}T \mbox{ when } \frac{N}{2^k}T\leq t< \frac{N+1}{2^k}T \mbox{ for some }N\in \mathbb{N}.$$

Consider the linear SDE
\begin{equation}
\begin{cases}
dZ(t)=[A(t)X(t)+C(t)\gamma(t)X(t)]dt+[-F(t)X(t)+D(t)\gamma(t)X(t)]dW(t), \\
Z(0)=x_0.
\end{cases}
\end{equation}
We now prove that
\begin{equation}\label{zstrong}
\begin{aligned}
\lim\limits_{n\to\infty}\int_0^T \E|V_n(t)-Z(t)|^2 dt = 0.
\end{aligned}
\end{equation}
To this end, we first analyze $V_n(t)-Z(t)$. We have
\begin{equation}
\begin{aligned}
&V_n(t)-Z(t)=\int_0^t \left[ A(u)(V_n(u)-X(u))+C(u)(U_n(u)-\gamma(u)X(u))\right]du\\
&+\int_0^t \left[ -F(u)(V_n(u)-X(u))+D(u)(U_n(u)-\gamma(u)X(u))\right] dW(u)\\
&=:Q_{1,n}(t)+Q_{2,n}(t).
\end{aligned}
\end{equation}
As a result,
\begin{equation}\label{op1}
\int_0^T \E|V_n(t)-Z(t)|^2 dt
\leq 2\int_0^T \E|Q_{1,n}(t)|^2 dt+2\int_0^T \E|Q_{2,n}(t)|^2 dt.
\end{equation}
We proceed to analyze $\E[Q_{1,n}^2(t)] $ and $\E[Q_{2,n}^2(t)] $, respectively. First
\begin{equation}\label{Q1P}
\begin{aligned}
\E|Q_{1,n}(t)|^2&\leq T\E\int_0^t \left|A(u)(V_n(u)-X(u))+C(u)(U_n(u)-\gamma(u)X(u))\right|^2du\\
&\leq 2TA^* \int_0^t \E|V_n(u)-X(u)|^2du+2TC^*\int_0^t \E|U_n(u)-\gamma(u)X(u)|^2du.
\end{aligned}
\end{equation}
By the strong convergence of $V_n$ to $X$, the first term above converges to $0$ as $n\to\infty$.  For the second term,
\begin{equation}\label{Q121}
\begin{aligned}
&\int_0^t \E|U_n(u)-\gamma(u)X(u)|^2du\\
&=\int_0^t \E|U_n(u)-\gamma(u)V_n(u)+\gamma(u)V_n(u)-\gamma(u)X(u)|^2du\\
&\leq 2(\gamma^{*})^2\int_0^t \E|V_n(u)-X(u)|^2du+2\int_0^t \E\left|\sum\limits_{k=n}^{N(n)} a_{k,n}\left[\gamma(u)X_k(u)-\gamma(m_{u,k})X_k(m_{u,k})\right]\right|^2 du.
\end{aligned}
\end{equation}
Now,
\begin{equation}\label{Q1n}
\begin{aligned}
&\int_0^t \E\left|\sum\limits_{k=n}^{N(n)} a_{k,n}\left[\gamma(u)X_k(u)-\gamma(m_{u,k})X_k(m_{u,k})\right]\right|^2 du \\
&=\int_0^t \E|\sum\limits_{k=n}^{N(n)} a_{k,n}(\gamma(u)-\gamma(m_{u,k}))X_k(u)+a_{k,n}\gamma(m_{u,k})(X_k(u)-X_k(m_{u,k}))|^2 du\\
&\leq 2\int_0^t \left[\E|\sum\limits_{k=n}^{N(n)} a_{k,n}(\gamma(u)-\gamma(m_{u,k}))X_k(u)|^2+\E|\sum\limits_{k=n}^{N(n)}a_{k,n}\gamma(m_{u,k})(X_k(u)-X_k(m_{u,k}))|^2 \right]du\\
&\leq 2\int_0^t \left[\sum\limits_{k=n}^{N(n)} a_{k,n}\E|(\gamma(u)-\gamma(m_{u,k}))X_k(u)|^2+\sum\limits_{k=n}^{N(n)}a_{k,n}
\E|\gamma(m_{u,k})(X_k(u)-X_k(m_{u,k}))|^2\right] du,
\end{aligned}
\end{equation}
where the last inequality follows from the  convexity of the function $f(x)=x^2$. Because $\gamma(\cdot)$ is continuous on $[0,T]$, it is uniformly continuous. Hence, for any $\varepsilon>0$ there is $n_0\in \mathbb{N}$ such that $|\gamma(t)-\gamma(s)|\leq \varepsilon$ whenever $t,s\in [0,T]$ with $|t-s|\leq \frac{1}{2^{n_0}}T$. For $n\geq n_0$, we then have 
\begin{equation}\label{Q1P2}
\begin{aligned}
&2\int_0^t \left[\sum\limits_{k=n}^{N(n)} a_{k,n}\E|(\gamma(u)-\gamma(m_{u,k}))X_k(u)|^2+\sum\limits_{k=n}^{N(n)}a_{k,n}
\E|\gamma(m_{u,k})(X_k(u)-X_k(m_{u,k}))|^2\right] du\\
&\leq 2\int_0^t\left[ \varepsilon^2 \max\limits_{n\leq k\leq N(n)} \E|X_k(u)|^2+(\gamma^{*})^2\max\limits_{n\leq k\leq N(n)}\E|X_k(u)-X_k(m_{u,k})|^2 \right]du\\
&\leq 2\int_0^t \left[ \varepsilon^2 Y(u)+(\gamma^{*})^2\frac{4T}{2^n}(A^*+\gamma^*C^* + \gamma^*D^* + F^*)Y(T)\right]du\\
&\leq 2\left[\varepsilon^2 +(\gamma^{*})^2\frac{4T}{2^n}(A^*+\gamma^*C^* + \gamma^*D^* + F^*)\right]TY(T),
\end{aligned}
\end{equation}
where the second inequality is by Lemma \ref{boundode} and the proof of Lemma \ref{lemmaclose}. Taking $n\to \infty$ and then $\varepsilon\to 0$, we obtain
\begin{equation}\label{Q122}
\lim\limits_{n\to \infty} \int_0^t \E|\sum\limits_{k=n}^{N(n)} a_k^n\left(\gamma(u)X_k(u)-\gamma(m_{u,k})X_k(m_{u,k})\right)|^2 du = 0.
\end{equation}
Combining (\ref{Q1P}), (\ref{Q121}) and (\ref{Q122}) yields
\begin{equation}
\begin{aligned}
\lim\limits_{n\to\infty}\E|Q_{1,n}(t)|^2 = 0.
\end{aligned}
\end{equation}
Moreover,  according to the above analysis the bound of $\E|Q_{1,n}(t)|^2$  does not depend on $t$; thus  the dominated convergence theorem gives
\begin{equation}\label{Q1F}
\lim\limits_{n\to \infty}\int_0^T \E|Q_{1,n}(t)|^2 dt=\int_0^T \lim\limits_{n\to \infty}\E|Q_{1,n}(t)|^2 dt=0.
\end{equation}

Employing It\^o's isometry we can derive similarly
\begin{equation}\label{Q2F}
\lim\limits_{n\to \infty}\int_0^T \E|Q_{2,n}(t)|^2 dt=\int_0^T \lim\limits_{n\to \infty}\E|Q_{2,n}(t)|^2 dt=0.
\end{equation}

By plugging (\ref{Q1F}) and (\ref{Q2F}) into (\ref{op1}) we establish (\ref{zstrong}), namely, $V_n \to Z\ \text{strongly in $L^2_{\mathcal{F}}([0,T];\mathbb{R})$}$.
Thus, $Z(t,\omega)=X(t,\omega)$ except on a zero measure set in the space of $[0,T]\times \Omega$. It follows that $X$ satisfies the same SDE as $Z$ or, equivalently, $X$ satisfies (\ref{XSDE}).
Moreover, it is immediate that this wealth equation is generated by the feedback policy
(\ref{PI}).
The proof is complete.

\newpage

\bibliographystyle{myplainnat}
\bibliography{Naive_MV}

\end{document}